\author{Martin Radloff}
\pdfoutput=1
	\documentclass[a4paper,12pt,oneside]{scrartcl}
	
\usepackage[sectionbib]{natbib}
\usepackage[]{hyperref}
\renewcommand{\theequation}{\thesection\arabic{equation}}
	\usepackage[english]{babel}   %andersherum, also "ngerman, american" sind die Ueberschriften in Englisch
	\usepackage[latin1]{inputenc}
	\usepackage[T1]{fontenc}
	\usepackage{ifthen}
	\usepackage{amsmath}
	\usepackage{amssymb}
	\usepackage{amsthm}
	\usepackage{latexsym, mathrsfs}
	\usepackage{bbm}
	\usepackage{geometry}
	\usepackage{graphicx}
	\usepackage{subfigure}
	\usepackage{xcolor}
\usepackage{pgf,tikz}
\usetikzlibrary{arrows}	
	\usepackage{multirow}
	\usepackage{multicol}
	\usepackage{enumerate, enumitem}	
	\usepackage{url}
	\usepackage{scrpage2}
	\usepackage{suetterl}	
\newfont{\suet}{suet14}
\newfont{\schwell}{schwell}
\DeclareTextFontCommand{\textsuet}{\suet}

\DeclareTextFontCommand{\textschwell}{\schwell}	

	\usepackage{datetime}
	\usepackage{blindtext}
	\defpagestyle{Kopfi}{
		{} {} {\small Martin Radloff, Rainer Schwabe\hfill $D$-optimal Designs for non-linear Models on Balls%\\ {\rule[-6pt]{200pt}{1pt}\hspace*{-3.25cm}}
		}
	}{
		{} {} {\hfill\pagemark}
	}

\newcommand{\BIGOP}[1]{\mathop{\mathchoice%
{\raise-0.22em\hbox{\huge $#1$}}%
{\raise-0.05em\hbox{\Large $#1$}}{\hbox{\large $#1$}}{#1}}}

\newcommand{\BIGboxplus}{\mathop{\mathchoice%
{\raise-0.35em\hbox{\huge $\boxplus$}}%
{\raise-0.15em\hbox{\Large $\boxplus$}}{\hbox{\large $\boxplus$}}{\boxplus}}}

\newtheorem{theorem}{Theorem}
\newtheorem{lemma}{Lemma}

\theoremstyle{definition}

\newtheorem{example}{Example}
\newtheorem{remark}{Remark}
\newtheorem{notation}{Notation}

\newcommand{\D}{\mathrm{d}}
\newcommand{\id}{\operatorname{id}}

\usepackage{dsfont}

\begin{document}
	\pagestyle{Kopfi}
	\thispagestyle{empty}
% Überschrift
	\vspace*{1.3cm}
	\begin{center}
		{\Large\textbf{Locally $D$-optimal Designs for Non-linear Models on the $k$-dimensional Ball}}
	\end{center}
	\centerline{Martin Radloff\footnote[2]{corresponding author: Martin Radloff, Institute for Mathematical Stochastics, Otto-von-Guericke-University, PF~4120, 39016~Magdeburg, Germany, \url{martin.radloff@ovgu.de}} and Rainer Schwabe\footnote[3]{Rainer Schwabe, Institute for Mathematical Stochastics, Otto-von-Guericke-University, PF~4120, 39016~Magdeburg, Germany, \url{rainer.schwabe@ovgu.de}}} 
		
	\vspace*{1.3cm}

%\begin{abstract}
  %%\textcolor[rgb]{1,0,0}{\blindtext}
	%In this paper we construct (locally) $D$-optimal designs for a wide class of non-linear multiple regression models, when the design region is a $k$-di\-men\-sion\-al ball. For this construction we make use of the concept of invariance and equivariance in the context of optimal designs. As examples we consider Poisson and negative binomial regression as well as proportional hazard models with censoring. By generalisation we can extend these results to arbitrary ellipsoids. 
%\end{abstract}
%
	%\vspace*{1.3cm}

\begin{quotation}
\noindent {\textit{Abstract:}}
%{\bf Contents of the Abstract.}\\
In this paper we construct (locally) $D$-optimal designs for a wide class of non-linear multiple regression models, when the design region is a $k$-di\-men\-sion\-al ball. For this construction we make use of the concept of invariance and equivariance in the context of optimal designs. As examples we consider Poisson and negative binomial regression as well as proportional hazard models with censoring. By generalisation we can extend these results to arbitrary ellipsoids.

\vspace{9pt}
\par
\end{quotation}\par

\def\thefigure{\arabic{figure}}
\def\thetable{\arabic{table}}

\renewcommand{\theequation}{\thesection.\arabic{equation}}

\fontsize{12}{14pt plus.8pt minus .6pt}\selectfont

\setcounter{equation}{0} %-1
%\noindent {\bf 1. Introduction}
\section{Introduction}
\label{sec:1}

	To find an optimal design, that means to find an optimal setting of control variables, of a special class of linear and non-linear models with respect to the $D$-criterion we will use results for equivariance and invariance in \cite{Radloff:2016}. So it is possible to reduce this multidimensional problem to a one-dimensional marginal problem. This marginal issue was investigated, for example in \cite{Konstantinou:2014}.
	The corresponding result for the linear case is well-known, see, for example, in \citet[Section 15.12]{Pukelsheim:1993}, and will be revisited in Section~\ref{sec:linMod}.\\	
	\cite{Schmidt:2017} considered the same class of models with $k$ covariates, but on a $k$-dimensional cuboid. They found a way to divide this problem into $k$ marginal sub-problems with only one covariate in the form like \cite{Konstantinou:2014}.\\	
	Our main result for non-linear models can be found in Section~\ref{sec:nonlinMod}. \\	
  In Section~\ref{sec3} we will discuss some examples. In the case of Poisson regression we will get a concrete formula to determine such an optimal design. In the case of negative binomial regression and censoring data models some computational efforts are needed.\\
	In the final Section~\ref{sec4} we have a short look at the generalisation to an ellipsoidal design region.

\setcounter{equation}{0} %-1
%\noindent {\bf 2. The Second Section}
\section{General Model Description, Information, and Design}
\label{sec:2}

	In the following sections we want to focus on a class of (non-linear) multiple regression models. Here every observation $Y$ depends on a special setting of control variables, a so-called design point $\boldsymbol{x}$. Each design point $\boldsymbol{x}$ is in the design region $\mathscr{X}=\mathbb{B}_k=\linebreak[2]\{\boldsymbol{x}\in\mathbb{R}^k\ :\ x_1^2+\ldots+x_k^2\leq 1\}$, the $k$-dimensional unit ball, $k\in\mathbb{N}$. The regression function $\boldsymbol{f}:\mathscr{X}\to\mathbb{R}^{k+1}$ is considered to be $\boldsymbol{x}\mapsto(1,x_1,\ldots,x_k)^\top$, and the parameter vector $\boldsymbol{\beta}=(\beta_0,\beta_1,\ldots,\beta_k)^\top$ is unknown and lies in the parameter space $\mathscr{B}$ which is assumed to be rotation invariant with respect to $\beta_1,\ldots,\beta_k$. We will only consider $\mathscr{B}=\mathbb{R}^{k+1}$. %but $\mathscr{B}=\mathbb{R}\times\mathbb{S}_{k-1}$ is also imaginable.
	And therefore the linear predictor is
	\[\boldsymbol{f}(\boldsymbol{x})^\top\boldsymbol{\beta} = \beta_0 + \beta_1 x_1 + \ldots + \beta_k x_k\ .\] 	
	A second requirement is that the one-support-point (or elemental -- as it is called, for example, in \cite{Atkinson:2014}) information matrix $\boldsymbol{M}(\boldsymbol{x},\boldsymbol{\beta})$ can be written in the form
	\begin{equation*}
		\boldsymbol{M}(\boldsymbol{x},\boldsymbol{\beta})=\lambda\!\left(\boldsymbol{f}(\boldsymbol{x})^\top\boldsymbol{\beta}\right)\boldsymbol{f}(\boldsymbol{x})\boldsymbol{f}(\boldsymbol{x})^\top
	\end{equation*}
	with an intensity (or efficiency) function $\lambda$ (see \citet[Section~1.5]{Fedorov:1972}) which only depends on the value of the linear predictor.
	
	\noindent In generalised linear models ({see \cite{McCulloch:2001}}) or for example in censored data models ({see \cite{Schmidt:2017}}) this prerequisite is fulfilled.\bigskip
	
	Now we want to find optimal designs on the the $k$-dimensional unit ball for those problems. This will be done in the sense of $D$-optimality, which is the most popular criterion and optimises the volume of the (asymptotic) confidence ellipsoid.\\
	For that account we need the concept of information matrices. In our case the information matrix of a (generalised) design $\xi$ with independent observations is
	\begin{equation*}
		\boldsymbol{M}(\xi,\boldsymbol{\beta})=\int_\mathscr{X}\boldsymbol{M}(\boldsymbol{x},\boldsymbol{\beta})\ \xi(\D \boldsymbol{x})=\int_\mathscr{X}\lambda\!\left(\boldsymbol{f}(\boldsymbol{x})^\top\boldsymbol{\beta}\right)\boldsymbol{f}(\boldsymbol{x})\boldsymbol{f}(\boldsymbol{x})^\top \xi(\D \boldsymbol{x})\ .
	\end{equation*}	
	Here generalised design does not only mean design on a discrete set of design points. It means an arbitrary probability measure on the design region, see, for example, \citet[Section~3.1]{Silvey:1980}.\\ 	
	So we can define: A design $\xi^\ast$ with regular information matrix $\boldsymbol{M}(\xi^\ast,\boldsymbol{\beta})$ is called (locally) $D$-optimal (at $\boldsymbol{\beta}$) if $\det(\boldsymbol{M}(\xi^\ast,\boldsymbol{\beta}))\geq\det(\boldsymbol{M}(\xi,\boldsymbol{\beta}))$ holds for all possible probability measures $\xi$ on $\mathscr{X}$. \bigskip
	
\begin{notation}
	While $\mathbb{S}_{d-1}$, $d\in\{2,3,4,\ldots\}$, describes the unit sphere, which is the surface of a $d$-dimensional unit ball $\mathbb{B}_d$, the symbol $\mathbb{S}_{d-1}(r)$ denotes the sphere with radius $r$, which is the surface of the $d$-dimensional ball $\mathbb{B}_d(r)$ with radius $r$.\\	
	Introducing notations we have to mention $\mathbb{O}_d$ the $d$-dimensional zero-vector, $\mathbb{O}_{d_1\times d_2}$ the $(d_1\times d_2)$-dimensional zero-matrix, $\mathds{1}_d$ the $d$-dimensional one-vector, $\mathbb{I}_d$ the $(d\times d)$-di\-men\-sion\-al identity matrix and $\id$ the identity function.\hfill$\Diamond$ %\bigskip	
\end{notation}

\setcounter{equation}{0} %-1
%\noindent {\bf 2. The Second Section}
\section{Linear Model}
\label{sec:linMod}
We first start with linear models which are well-investigated in the literature. Here the intensity (or efficiency) function $\lambda$ is constant 1 and does not depend on the parameter $\boldsymbol{\beta}$. Hence, the information matrix and the $D$-optimal design is independent from $\boldsymbol{\beta}$. In \citet[Section 15.12]{Pukelsheim:1993}, for example, the following a little bit adapted result can be found:

\begin{theorem}\label{theorem:lin1}
	In the linear case with regression function \[\boldsymbol{f}(\boldsymbol{x})=(1,x_1,\ldots,x_k)^\top\] the vertices of an arbitrarily rotated $k$-dimensional regular simplex, whose vertices lie on the surface of the design region $\mathbb{S}_{k-1}$, constitutes a $D$-optimal design on the unit ball $\mathbb{B}_k$. 
	The corresponding information matrix is the diagonal matrix 
	\begin{align*}%\label{eq:InfoMatrix}
		\mathrm{diag}(1,\tfrac{1}{k},\ldots,\tfrac{1}{k})\ .
	\end{align*}
\end{theorem}

Here ``regular'' means that all edges of the simplex have the same length.

\begin{lemma}\label{lem:lin1}
	The (continuously) uniform design %(probability distribution) 
	on $\mathbb{S}_{k-1}$ has the same information matrix.%~\eqref{eq:InfoMatrix}.
\end{lemma}

\begin{proof}
	Let $\xi$ be the uniform design (or better: uniform probability measure) on %the $k$-dimensional sphere 
	$\mathbb{S}_{k-1}$.\\ 
	We start with $k\geq 2$. For all $i\in\{1,2,\ldots,k\}$ let $(X_1,\ldots,X_i,\ldots,X_k)\sim\xi$. It follows $(X_1,\ldots,-X_i,\ldots,X_k),$ $(X_i,\ldots,X_1,\ldots,X_k)\sim\xi$. So we have $\mathrm{E}(X_i)=\mathrm{E}(-X_i)$ and $\mathrm{E}(X_i)=0$ for all $i$. $\mathrm{E}(X_i X_j)=\mathrm{E}((-X_i) X_j)$ and $\mathrm{E}(X_i X_j)=0$ for all $i\neq j$. It is $\mathrm{Var}(X_1)=\mathrm{Var}(X_i)$. With $\mathrm{E}(X_i)=0$ we get $\mathrm{E}(X_1^2)=\mathrm{E}(X_i^2)$ for all $i$. It is $k\,\mathrm{E}(X_1^2)=\mathrm{E}(X_1^2)+\ldots+\mathrm{E}(X_k^2)=\mathrm{E}(X_1^2+\ldots+X_k^2)=\mathrm{E}(1)=1$. Hence, $\mathrm{E}(X_1^2)=\ldots=\mathrm{E}(X_k^2)=\frac{1}{k}$.
	So it is obvious
	\begin{align*}
		\boldsymbol{M}(\xi)&%=\int_{\mathbb{B}_k} \boldsymbol{f}(\boldsymbol{x}) \boldsymbol{f}(\boldsymbol{x})^\top \xi(\D\boldsymbol{x})
		= \int_{\mathbb{S}_{k-1}} (1, x_1, \ldots, x_k)^\top (1, x_1, \ldots, x_k)\ \xi(\D\boldsymbol{x})
		= \mathrm{diag}(1,\tfrac{1}{k},\ldots,\tfrac{1}{k})\ .
	\end{align*}
	For $k=1$ the sphere $\mathbb{S}_{0}$ and the vertices of the simplex $[-1,1]$ are the same. So the information matrices coincide.			
\end{proof}

Hence, the uniform design is also $D$-optimal.

\setcounter{equation}{0} %-1
%\noindent {\bf 2. The Second Section}
\section{Non-linear Models}
\label{sec:nonlinMod}
In this section we want to develop our main results. Invariance and equivariance (see \cite{Radloff:2016}) help to reduce the complexity of this endeavour.

\begin{lemma} \label{L2} %Lemma 2 %Lemma 1 
	%\textcolor[rgb]{0,0,1}{(Schmidt, Schwabe)}
	A (locally) $D$-optimal design is concentrated on the surface of $\mathscr{X}=\mathbb{B}_k$ and is equivariant with respect to rotations.
\end{lemma}

\begin{remark}
	Equivariance in this context means: If the design or design region is rotated, the parameter space must be rotated in a corresponding way. This corresponding rotation is specified in the following proof.\hfill$\Diamond$
\end{remark}

\begin{proof}
	Every information matrix of a design with an inner point as a support point can be majorised (in the positive semidefinite sense) by an information matrix which is only defined on the surface of the ball. Because of the strict convexity of the $D$-criterion we only have to look at the surface.\\
	In the context of rotational equivariance the statements in \cite{Radloff:2016} are specialised:\\
	If $\boldsymbol{g}$ is a rotation (a special one-to-one-mapping) of the design region $\mathscr{X}=\mathbb{B}_k$ with rotation matrix $\boldsymbol{R}_{\boldsymbol{g}}$, so that \mbox{$\boldsymbol{g}(\boldsymbol{x})=\boldsymbol{R}_{\boldsymbol{g}}\boldsymbol{x}$}, then there exists an orthogonal $(k+1)\times(k+1)$-matrix $\boldsymbol{Q}_{\boldsymbol{g}}$ with determinant 1, namely 
	\[\boldsymbol{Q}_{\boldsymbol{g}}= \begin{pmatrix} 1 &\mathbb{O}_k^\top \\ \mathbb{O}_k&\boldsymbol{R}_{\boldsymbol{g}}\end{pmatrix},\]
	such that $\boldsymbol{f}(\boldsymbol{g}(\boldsymbol{x}))=\boldsymbol{Q}_{\boldsymbol{g}}\boldsymbol{f}(\boldsymbol{x})$.
	%To have an equivariance the parameter must be transformed as well. 
	With the corresponding rotation $\boldsymbol{\tilde{g}}(\boldsymbol{\beta})=(\boldsymbol{Q}_{\boldsymbol{g}}^\top)^{-1}\boldsymbol{\beta}$, which is $\boldsymbol{\tilde{g}}(\boldsymbol{\beta})=\boldsymbol{Q}_{\boldsymbol{g}}\boldsymbol{\beta}$ because of orthogonality, we have $\boldsymbol{f}(\boldsymbol{g}(\boldsymbol{x}))^\top \boldsymbol{\tilde{g}}(\boldsymbol{\beta})=\boldsymbol{f}(\boldsymbol{x})^\top\boldsymbol{\beta}$. And, if $\xi$ is a (locally) $D$-optimal design for $\boldsymbol{\beta}$, then $\xi^{\boldsymbol{g}}:=\xi\circ\boldsymbol{g}^{-1}$ is (locally) $D$-optimal for $\boldsymbol{\tilde{g}}(\boldsymbol{\beta})$. 	
\end{proof}

\begin{remark} %Corollary 1
	If $(\beta_1,\ldots,\beta_k)^\top\neq\mathbb{O}_k$, there is a rotation $\boldsymbol{\tilde{g}}$ such that $\boldsymbol{\tilde{g}}(\beta_0,\beta_1,\ldots,\beta_k)^\top\linebreak[1]=(\beta_0,\tilde{\beta}_1,0,\ldots,0)$ with $\tilde{\beta}_1=||(\beta_1,\ldots,\beta_k)^\top||>0$, where $||\cdot||$ is the ($k$-dimensional) Euclidean norm. If $(\beta_1,\ldots,\beta_k)^\top=\mathbb{O}_k$, then no rotation is needed. So only the case $\boldsymbol{\beta}\in\mathscr{B}$ with
	\begin{equation}
		\label{eq:speccase}
			\beta_1\geq 0, \beta_2=\ldots=\beta_k=0
	\end{equation}
	has to be considered.
	The search for a (locally) optimal design with an initial guess of the parameter vector in the whole parameter space $\mathscr{B}$ reduces to only the length of this vector.\hfill$\Diamond$
\end{remark}

The next results mostly need~\eqref{eq:speccase} and $\beta_1>0$. The case $\beta_1=0$ will be discussed at the end of this section in Remark~\ref{rem1}.

\begin{lemma} %Lemma 3 
\label{lem3}
	For $\boldsymbol{\beta}$ satisfying~\eqref{eq:speccase} the $D$-criterion is invariant with respect to rotations of $x_2,\ldots,x_k$.
\end{lemma}

\begin{proof}
	%The statements in {Radloff and Schwabe \cite{Radloff:2016}} are specialised on this context.\\
	Analogously to Lemma~\ref{L2}, if $\boldsymbol{g}$ is a rotation of $x_2,\ldots,x_k$, then there exists a \mbox{$(k-1)\times(k-1)$}-rotation matrix $\boldsymbol{R}_{\boldsymbol{g}}$, such that 
	\[\boldsymbol{Q}_{\boldsymbol{g}}= \begin{pmatrix} \text{\scriptsize$\begin{matrix}1&\\&1\end{matrix}$}&\mathbb{O}\\ \mathbb{O}&\boldsymbol{R}_{\boldsymbol{g}}\end{pmatrix}\]
	is orthogonal with determinant 1 and $\boldsymbol{f}(\boldsymbol{g}(\boldsymbol{x}))=\boldsymbol{Q}_{\boldsymbol{g}}\boldsymbol{f}(\boldsymbol{x})$. Hence, $\boldsymbol{\tilde{g}}(\boldsymbol{\beta})=\boldsymbol{Q}_{\boldsymbol{g}} \boldsymbol{\beta}$ and for all rotations $\boldsymbol{g}$ of $x_2,\ldots,x_k$ we have $\boldsymbol{\tilde{g}}(\boldsymbol{\beta})=\boldsymbol{\beta}$ for all $\boldsymbol{\beta}$ in~\eqref{eq:speccase}. And so in notation of \cite{Radloff:2016}
for all $\boldsymbol{\beta}$ satisfying~\eqref{eq:speccase} and rotations $\boldsymbol{g}$ of $x_2,\ldots,x_k$ we get
	\begin{equation*} 
		\det[\boldsymbol{M}(\xi^{\boldsymbol{g}},\boldsymbol{\beta})]
		=\det[\boldsymbol{M}(\xi^{\boldsymbol{g}},\boldsymbol{\tilde{g}}(\boldsymbol{\beta}))]
		=\det[\boldsymbol{Q}_{\boldsymbol{g}}\boldsymbol{M}(\xi,\boldsymbol{\beta})\boldsymbol{Q}_{\boldsymbol{g}}^\top]
		=\det[\boldsymbol{M}(\xi,\boldsymbol{\beta})] \ . \qedhere
	\end{equation*}
\end{proof}

So we can find an optimal design within the class of invariant designs on the surface of the ball. The concept of marginal and conditional designs (see \cite{Cook:1980}) can be used.

\begin{lemma} \label{lem4} %Lemma 4
	For $\boldsymbol{\beta}$ satisfying~\eqref{eq:speccase} the invariant designs (on the surface) with respect to rotations of $x_2,\ldots,x_k$ are given by $\xi_1\otimes\overline{\eta}$, where $\xi_1$ is a marginal design on $[-1,1]$ and $\overline{\eta}$ is a Markov kernel (conditional design). For fixed $x_1$ the kernel $\overline{\eta}(x_1,\cdot)$ is the uniform distribution on the surface of a $(k-1)$-dimensional ball with radius $\sqrt{1-x_1^2}$.
\end{lemma}

\begin{remark} 
	If $x_1\in\{-1,1\}$, the $(k-1)$-dimensional ball with the uniform distribution is degenerated as a point. So it is only a one-point-measure.\hfill$\Diamond$
\end{remark}

\begin{proof}
	Each design / probability measure $\xi$ on a $k$-dimensional Borel set $(\mathbb{R}^k,\mathcal{B}(\mathbb{R}^k))=(\mathbb{R}\times\mathbb{R}^{k-1},\mathcal{B}(\mathbb{R})\otimes\mathcal{B}(\mathbb{R}^{k-1}))$ can be split into a marginal probability measure $\xi_1$ on $(\mathbb{R},\mathcal{B}(\mathbb{R}))$ with $\xi_1(A)=\xi(A\times\mathbb{R}^{k-1})$ and a kernel $\overline{\eta}$ with source $(\mathbb{R},\mathcal{B}(\mathbb{R}))$ and target $(\mathbb{R}^{k-1},\mathcal{B}(\mathbb{R}^{k-1}))$ (unique up to sets of measure zero), see, for example, \citet[Section 8.3]{Klenke:2014}. Hence, $\xi=\xi_1\otimes\overline{\eta}$.
	
	\noindent We only want to focus on designs on $\mathbb{S}_{k-1}$. %the surface of a $k$-dimensional ball.
	So the domains of these measures and kernels can be restricted. We have $\xi(\mathbb{S}_{k-1})=1$, $\xi_1([-1,1])=1$ and $\overline{\eta}(x_1,\mathbb{S}_{k-2}({\textstyle\sqrt{1-x_1^2}}))=1$ for all $x_1\in[-1,1]$.
	
	\noindent The design $\xi$ should be invariant with respect to rotations of $x_2,\ldots,x_k$. So for all $x_1\in[-1,1]$ the probability measures $\overline{\eta}(x_1,\cdot)$ have to be invariant, too. The group of all rotations of $x_2,\ldots,x_k$ is a locally compact group, so that the Haar-probability-measure is unique (see \citet[\S 60]{Halmos:1974}). The uniform distribution on $\mathbb{S}_{k-2}(\sqrt{1-x_1^2})$ is such an invariant measure. Hence, $\overline{\eta}(x_1,\cdot)$ must be uniform.
\end{proof}

\begin{lemma} %Lemma 4a 
\label{L4a}
	The information matrix for $\xi_1\otimes\overline{\eta}$ in Lemma~\ref{lem4} is
\renewcommand*\arraystretch{1.2}
	\begin{equation}
			\boldsymbol{M}(\xi_1\otimes\overline{\eta})=
					\left(\begin{array}{c|c}
							\hspace*{-0.5em}\begin{array}{cc}
							\int q\,\D\xi_1 		 & \int q \id \D\xi_1\\
							\int q \id \D\xi_1 & \int q \id^2 \D\xi_1
							\end{array} & \mathbb{O}_{2\times (k-1)}\\ \hline
							\mathbb{O}_{(k-1)\times 2} & \frac{1}{k-1} \int q\,(1-\id^2)\,\D\xi_1\ \mathbb{I}_{k-1}							
					\end{array} \right)
	\label{eq:infomatrix}
	\end{equation}
	with $q(x_1):=\lambda(\beta_0+\beta_1 x_1)$.
\renewcommand*\arraystretch{1}
\end{lemma}	

\begin{proof} Let $\boldsymbol{\tilde{x}}=(x_2,\ldots,x_k)$. We have to determine $\int_{\mathbb{S}_{k-1}} q(x_1)\, x_i^\kappa x_j^\nu\,$ $(\xi_1\otimes\overline{\eta})(\D(x_1,\boldsymbol{\tilde{x}}))$ for $i,j\in\{1,\ldots,k\}$ and $\kappa,\nu\in\{0,1\}$. Remembering that $\overline{\eta}(x_1,\cdot)$ is uniformly distributed, we use symmetry properties as in the proof of Lemma~\ref{lem:lin1} to do the calculations.
\end{proof}

To use the Kiefer-Wolfowitz Equivalence Theorem for $D$-optimality we need the structure of the sensitivity function $$\psi(\boldsymbol{x},\xi_1\otimes\overline{\eta}) = \lambda\!\left(\boldsymbol{f}(\boldsymbol{x})^\top\boldsymbol{\beta}\right)\boldsymbol{f}(\boldsymbol{x})^\top\boldsymbol{M}^{-1}\left(\xi_1\otimes\overline{\eta}\right)\boldsymbol{f}(\boldsymbol{x})\ .$$

\begin{lemma} %Lemma 5 and Lemma 5a 
\label{lem5}
	For the invariant designs $\xi_1\otimes\overline{\eta}$ with respect to rotations of $x_2,\ldots,x_k$ in~\eqref{eq:speccase} the sensitivity function $\psi$ is invariant (constant on orbits) and has for $\boldsymbol{x}\in\mathbb{S}_{k-1}$ the form
	\begin{equation}
		\psi(\boldsymbol{x},\xi_1\otimes\overline{\eta})=q(x_1)\cdot p_1(x_1) \quad\text{with}\quad \boldsymbol{x}=(x_1,\ldots,x_k)^\top
	\label{eq:5a}
	\end{equation}
	where $p_1$ is a polynomial of degree 2 in $x_1$.	
\end{lemma}

\begin{proof}
With $D:=\int q\,\D\xi_1\ \int q \id^2 \D\xi_1-\left(\int q \id \D\xi_1\right)^2$ we get from Lemma~\ref{L4a} 
\[\boldsymbol{M}^{-1}(\xi_1\otimes\overline{\eta}) =
\renewcommand*\arraystretch{1.4}
					\left(\begin{array}{c|c}
							\begin{array}{cc}
							\frac{1}{D}\int q \id^2 \D\xi_1 		 & -\frac{1}{D}\int q \id \D\xi_1\\
							-\frac{1}{D}\int q \id \D\xi_1 & \frac{1}{D}\int q\,\D\xi_1
							\end{array} & \mathbb{O}_{2\times (k-1)}\\ \hline
							\mathbb{O}_{(k-1)\times 2} & \frac{k-1}{\int q\,(1-\id^2)\,\D\xi_1}\ \mathbb{I}_{k-1}							
					\end{array} \right)
\renewcommand*\arraystretch{1}
\]
and with $x_2^2+\ldots+x_k^2=1-x_1^2$
\begin{multline*}
	\psi(\boldsymbol{x},\xi_1\otimes\overline{\eta})	
	=q(x_1) \left[  \frac{1}{D} \left(\int\! q \id^2 \D\xi_1 - 2x_1 \int\! q \id \D\xi_1 + x_1^2 \int\! q\, \D\xi_1 \right)  + \frac{(k-1) (1-x_1^2)}{\int q\,(1-\id^2)\,\D\xi_1}  \right] . %\qedhere
\end{multline*}
If $k=1$, the diagonal matrix part $\frac{k-1}{\int q\,(1-\id^2)\,\D\xi_1}\ \mathbb{I}_{k-1}$ of the inverted information matrix and the second summand $\frac{(k-1) (1-x_1^2)}{\int q\,(1-\id^2)\,\D\xi_1}$ in the sensitivity function vanish.
\end{proof}

The intensity function is now assumed to satisfy the following four conditions {(see \cite{Konstantinou:2014} or \cite{Schmidt:2017})}:

\begin{enumerate}
	\item[(A1)] $\lambda$ is positive on $\mathbb{R}$ and twice continuously differentiable.
	\item[(A2)] $\lambda^\prime$ is positive on $\mathbb{R}$.
	\item[(A3)] The second derivative $u^{\prime\prime}$ of $u=\frac{1}{\lambda}$ is injective on $\mathbb{R}$.
	\item[(A4)] The function $\frac{\lambda^\prime}{\lambda}$ is a non-increasing function.
\end{enumerate}

\begin{remark}
	In \cite{Konstantinou:2014} or \cite{Schmidt:2017} the assumption (A4) looks a little bit different. There the function $\frac{\lambda}{\lambda^\prime}$ should be non-decreasing. But both statements are equivalent if we postulate (A1) and (A2).\hfill$\Diamond$
\end{remark}

\begin{lemma}
	If the intensity function $\lambda$ satisfies the conditions (A1), (A2), (A3) and (A4), then $q$ as defined in Lemma~\ref{L4a} with $\beta_1>0$ provides the same properties, respectively.
\end{lemma}

\begin{example} \label{ex1}
	In Poisson regression every observation $Y_i$ in $\boldsymbol{x}_i$ is Poisson distributed with $\mathrm{E}(Y_i)=\exp(\boldsymbol{f}(\boldsymbol{x}_i)^\top\boldsymbol{\beta})$. The (transformed) intensity function is $q_\mathrm{P}$.\\
	A generalisation of the Poisson regression is the negative binomial regression (or Poisson-gamma regression). Every observation $Y_i$ in $\boldsymbol{x}_i$ is negative-binomially distributed with $\mathrm{E}(Y_i)=\mu_i :=\exp(\boldsymbol{f}(\boldsymbol{x}_i)^\top\boldsymbol{\beta})$ and $\mathrm{Var}(Y_i)=\mu_i+a\mu_i^2$ for a fixed $a\geq 0$. The (transformed) intensity function is $q_\mathrm{NB}$.\\
	Both regression models satisfy (A1), (A2), (A3) and (A4) for $\beta_1>0$:
	\begin{align*}
		q_\mathrm{P}(x_1)&=\exp(\beta_0+\beta_1 x_1)&q_\mathrm{NB}(x_1)&=\frac{\exp(\beta_0+\beta_1 x_1)}{1+a\exp(\beta_0+\beta_1 x_1)}\\
		q_\mathrm{P}^\prime(x_1)&=\beta_1\,\exp(\beta_0+\beta_1 x_1)&q_\mathrm{NB}^\prime(x_1)&=\frac{\beta_1\,\exp(\beta_0+\beta_1 x_1)}{(1+a\exp(\beta_0+\beta_1 x_1))^2}\\
		\frac{q_\mathrm{P}^\prime(x_1)}{q_\mathrm{P}(x_1)}&=\beta_1&\frac{q_\mathrm{NB}^\prime(x_1)}{q_\mathrm{NB}(x_1)}&=\frac{\beta_1}{1+a\exp(\beta_0+\beta_1 x_1)}\\		
		u_\mathrm{P}(x_1)&=\exp(-\beta_0-\beta_1 x_1)&u_\mathrm{NB}(x_1)&=\frac{1+a\exp(\beta_0+\beta_1 x_1)}{\exp(\beta_0+\beta_1 x_1)}\\
		u_\mathrm{P}^{\prime\prime}(x_1)&=\beta_1^2\,\exp(-\beta_0-\beta_1 x_1)&u_\mathrm{NB}^{\prime\prime}(x_1)&=\beta_1^2\,\exp(-\beta_0-\beta_1 x_1) \tag*{$\Diamond$}%\\
	\end{align*}
\end{example}

\begin{lemma} \label{L6} %Lemma 6
	In~\eqref{eq:speccase}: If $q$ satisfies (A1), (A2) and (A3), then the (locally) $D$-optimal marginal design $\xi_1^\ast$ is concentrated on exactly 2 points $x_{11}^\ast, x_{12}^\ast$.
\end{lemma}

\begin{proof}
  By the Kiefer-Wolfowitz Equivalence Theorem for $D$-optimality we have to proof 
	\[k+1\geq\psi(\boldsymbol{x},\xi_1\otimes\overline{\eta})=q(x_1)\cdot p_1(x_1)\quad\text{for all}\quad\boldsymbol{x}=(x_1,\ldots,x_k)^\top .\]
	This is equivalent to 
	\begin{equation}
		\frac{p_1(x_1)}{k+1}\leq \frac{1}{q(x_1)}\ .
	\label{eq:IneqEquivalenceThBiedermann}
	\end{equation}
	The proof then follows along the same lines as in \cite{Konstantinou:2014} where they discuss in the proof of Lemma~1 a similar inequality.
	This gives us the fact, that there are at most 2 points.\\
	Assume, that $\xi_1$ has only 1 support point. So $D$ in the proof of Lemma~\ref{lem5} would be 0 and the inverse of the information matrix and thus the polynomial $p_1$ would not exist. Contradiction. Hence, $\xi_1$ has exactly 2 points.
\end{proof}

The next lemma characterises these 2 points and their weights, while $x_{12}^\ast$ is specified in Lemma~\ref{L9}.

\begin{lemma}\label{L8}% Lemma 8
	In the settings of Lemma~\ref{L6} a potential (locally) $D$-optimal marginal design $\xi_1$ has the weights 
	\begin{equation}
		\xi_1^\ast(x_{11}^\ast)=\frac{1}{k+1}, \qquad \xi_1^\ast(x_{12}^\ast)=\frac{k}{k+1}
	\end{equation} 
	where $x_{11}^\ast=1$ and $x_{12}^\ast\in[-1,1)$, where $x_{12}^\ast\neq -1$ for $k\geq 2$.
\end{lemma}

\begin{proof}
	With $\xi_1^\ast(x_{11}^\ast)=\alpha>0$ and $\xi_1^\ast(x_{12}^\ast)=1-\alpha$ it is
	\begin{equation*}
		\int q \id^\kappa \D\xi_1^\ast = q(x_{11}^\ast) x_{11}^{\ast\ \kappa} \alpha + q(x_{12}^\ast) x_{12}^{\ast\ \kappa} (1-\alpha) \quad,\quad \kappa\in\{0,1,2\}\ ,
	\end{equation*}
	and
	\begin{equation*}
		\int q\,\D\xi_1 \int q \id^2 \D\xi_1-\left(\int q \id \D\xi_1\right)^2 = q(x_{11}^\ast)q(x_{12}^\ast) (x_{11}^\ast-x_{12}^\ast)^2 \alpha (1-\alpha)\ .
	\end{equation*}
	By using this and the formulas in the proof of Lemma~\ref{lem5} we can determine the polynomial~$p_1$.
	\begin{align*}
		p_1(x_1)=&\frac{ q(x_{11}^\ast) (x_{11}^{\ast}-x_1)^2 \alpha + q(x_{12}^\ast) (x_{12}^{\ast}-x_1)^2 (1-\alpha) }{q(x_{11}^\ast)q(x_{12}^\ast) (x_{11}^\ast-x_{12}^\ast)^2 \alpha (1-\alpha)}\\
		&+ \frac{(k-1) (1-x_1^2)}{q(x_{11}^\ast) (1-x_{11}^{\ast\ 2}) \alpha + q(x_{12}^\ast) (1-x_{12}^{\ast\ 2}) (1-\alpha)}	
	\end{align*}
	Assume, that $x_{11}^\ast=1$. If we can find $\alpha$ and $x_{12}^\ast\in[-1,1)$ so that $\xi_1$ is a feasible (locally) $D$-optimal marginal design, we are done. Here we have to notice that for $k\geq 2$ it is $x_{12}^{\ast}\neq -1$, otherwise a~2-point-design $\xi$ to estimate $k+1$ parameters would exist.\\ If $k=1$, the second summand in $p_1$ as remarked in Lemma~\ref{lem5} is missing. So there is no division by zero in the second summand.\\	
	We look back to the inequality~\eqref{eq:IneqEquivalenceThBiedermann} of the Equivalence Theorem $p_1(x_1)\leq \frac{k+1}{q(x_1)}$.
	In $x_1=x_{11}^\ast\ (=1)$ there should be equality of this inequality:
	\begin{equation*}
		p_1(1)=\frac{1}{q(1) \alpha}\overset{!}{=} \frac{k+1}{q(1)} 	
	\end{equation*}	
	So $\alpha=\frac{1}{k+1}$ and consequently $1-\alpha=\frac{k}{k+1}$. With 	
	\begin{equation*}
		p_1(x_{12}^\ast)= \frac{ 1 }{q(x_{12}^\ast) (1-\alpha)} + \frac{(k-1) }{q(x_{12}^\ast) (1-\alpha)} = \frac{k+1}{q(x_{12}^\ast)}
	\end{equation*}
	there is equality of the inequality~\eqref{eq:IneqEquivalenceThBiedermann} in $x_1=x_{12}^\ast$, too.
\end{proof}

\begin{remark}
	In anticipation of Theorem~\ref{Theorem1} the discretised design will consist of $k+1$ equally weighted support points, where the uniform distribution in the $x_{12}^\ast$-hyperplane is substituted by $k$ points analogously as in Theorem~\ref{theorem:lin1} while the information matrix leaves unchanged. The weights in Lemma~\ref{L8} allow this discretisation.\hfill$\Diamond$ 
\end{remark}

\begin{lemma}\label{L9} %Lemma 9
	In the settings of Lemma~\ref{L8} for $k\geq 2\ :\ x_{12}^\ast\in(-1,1)$ is solution of 
	\begin{equation}
		\frac{q^\prime(x_{12}^\ast)}{q(x_{12}^\ast)}=\frac{2\,(1+kx_{12}^\ast)}{k\,(1-x_{12}^{\ast\ 2})}
	\label{eq:L9}
	\end{equation}
	and for $k=1$\ :\ It is $x_{12}^\ast=-1$ or $x_{12}^\ast\in[-1,1)$ is solution of
	\begin{equation}
		\frac{q^\prime(x_{12}^\ast)}{q(x_{12}^\ast)}=\frac{2}{1-x_{12}^{\ast}}\ .
	\label{eq:L92}
	\end{equation}
	In any case, if additionally (A4) is satisfied, the solution $x_{12}^\ast$ is unique. 
\end{lemma}

\begin{proof}
	With $x_{11}^\ast=1$, $\xi_1^\ast(x_{11}^\ast)=\frac{1}{k+1}$, $\xi_1^\ast(x_{12}^\ast)=\frac{k}{k+1}$ and the notation from Lemma~\ref{L8} the determinant of the information matrix for $k\geq 2$ is
	\begin{multline*}
		\det \boldsymbol{M}(\xi_1\otimes\overline{\eta}) = \left(\int q\,\D\xi_1 \int q \id^2 \D\xi_1-\left(\int q \id \D\xi_1\right)^2\right) \cdot\left(\frac{1}{k-1} \int q\,(1-\id^2)\,\D\xi_1\right)^{k-1}
	\end{multline*}
	and
	\begin{align*}
		\log \det \boldsymbol{M}(\xi_1\otimes\overline{\eta}) = &\log\!\left(q(1)q(x_{12}^\ast) (1-x_{12}^\ast)^2 \tfrac{k}{(k+1)^2}\right)\\ &+ (k-1) \left(\log\!\left( q(x_{12}^\ast) (1-x_{12}^{\ast\ 2}) \tfrac{k}{k+1} \right)-\log(k-1)\right)\ .
	\end{align*}	
	To optimise this we have to solve
	\begin{multline*}
		0 \stackrel{!}{=} \frac{\D}{\D x_{12}^\ast}\log \det \boldsymbol{M}(\xi_1\otimes\overline{\eta}) =
				\frac{q^\prime(x_{12}^\ast)}{q(x_{12}^\ast)}+\frac{-2\,(1-x_{12}^\ast)}{(1-x_{12}^\ast)^2}
				 +(k-1) \left(\frac{q^\prime(x_{12}^\ast)}{q(x_{12}^\ast)}+\frac{-2x_{12}^\ast}{1-x_{12}^{\ast\ 2}}\right)  \ .
	\end{multline*}		
	With $1-x_{12}^\ast\neq 0$ and $1+x_{12}^\ast\neq 0$ it simplifies to~\eqref{eq:L9}.\\
	The right-hand side $\frac{2\,(1+kx_{12}^\ast)}{k\,(1-x_{12}^{\ast\ 2})}$ of~\eqref{eq:L9} has poles in $-1$ and $1$ and is strictly increasing on $(-1,1)$ with a range of $(-\infty,\infty)$. Because of (A1) and (A2) $\frac{q^\prime}{q}$ is continuous on $\mathbb{R}$. So there must be an intersection. If $\frac{q^\prime}{q}$ is non-increasing (A4), the intersection is unique.  
	
	\noindent For $k=1$ the third summand of $\log \det \boldsymbol{M}(\xi_1\otimes\overline{\eta})$ disappears as mentioned. So we solve
	\[\frac{\D}{\D x_{12}^\ast}\log \det \boldsymbol{M}(\xi_1\otimes\overline{\eta}) 
				= \frac{q^\prime(x_{12}^\ast)}{q(x_{12}^\ast)}+\frac{-2\,(1-x_{12}^\ast)}{(1-x_{12}^\ast)^2}
				\stackrel{!}{=} 0\ .\]
	With $1-x_{12}^{\ast}\neq 0$ it simplifies to~\eqref{eq:L92}. 
	If~\eqref{eq:L92} has no solution in $[-1,1)$, the maximum of $\log \det \boldsymbol{M}(\xi_1\otimes\overline{\eta})$ is on the boundary. This is equivalent to the maximisation of \mbox{$q(x_{12}^\ast) (1-x_{12}^\ast)^2$} on the boundary. $x_{12}^{\ast}$ must be $-1$.\\
	The right-hand side of~\eqref{eq:L92}, $\frac{2}{1-x_{12}^{\ast}}$, is strictly increasing on $[-1,1)$ with values covering~\mbox{$[1,\infty)$}. If $\frac{q^\prime(-1)}{q(-1)}<1$, there cannot be a solution of~\eqref{eq:L92}. So $x_{12}^{\ast}$ is $-1$, unique. If~$\frac{q^\prime(-1)}{q(-1)}\geq 1$ and (A4) is satisfied, the solution of~\eqref{eq:L92} is unique. 
\end{proof}

Now we can discretise the found (generalised) design.

\begin{theorem} \label{Theorem1}
	There is a (locally) $D$-optimal design for the considered problem satisfying $\beta_1>0, \beta_2=\ldots=\beta_k=0,\ \beta_0\in\mathbb{R}$ that has one support point in $(1,0,\ldots,0)^\top$ and the other $k$~support points are the vertices of an arbitrarily rotated, $(k-1)$-dimensional simplex which is maximally inscribed in the intersection of the $k$-dimensional unit ball and a hyperplane with $x_1=x_{12}^\ast$ in Lemma~\ref{L9}. %\\
	The design is equally weighted with $\frac{1}{k+1}$.
\end{theorem}

\begin{proof}
	If $k=1$, the (locally) $D$-optimal (generalised) design $\xi^\ast=\xi_1^\ast\otimes\overline{\eta}$ consists of 2~points with weights $\frac{1}{2}$. This is discrete.\\
	Otherwise one explicit point is the pole $\boldsymbol{s}_0=(1,0,\ldots,0)^\top$ with weight $\frac{1}{k+1}$. So only the uniform distribution on the orbit $\{x_{12}^\ast\}\times\mathbb{S}_{k-2}(\sqrt{1-x_{12}^{\ast\ 2}})$ must be discretised.\\
	Being $\boldsymbol{\tilde{x}}=(x_2,\ldots,x_k)$ we consider the information matrix:
	\begin{align*}
		\boldsymbol{M}(\xi^\ast)&=\int_{\mathbb{B}_d} \lambda(\boldsymbol{f}(\boldsymbol{x})^\top\boldsymbol{\beta}) \boldsymbol{f}(\boldsymbol{x}) \boldsymbol{f}(\boldsymbol{x})^\top \xi^\ast(\D\boldsymbol{x})\\
		&=\frac{q(1)}{k+1} \,  \boldsymbol{f}(\boldsymbol{s}_0) \boldsymbol{f}(\boldsymbol{s}_0)^\top +
			\frac{k\, q(x_{12}^\ast)}{k+1} \hspace*{-2.2em}\int\limits_{\mathbb{S}_{k-2}\left(\sqrt{1-x_{12}^{\ast\ 2}}\right)} \hspace*{-2.2em}  \boldsymbol{f}\!\left(\!\!\binom{x_{12}^\ast}{\boldsymbol{\tilde{x}}}\!\!\right) \boldsymbol{f}\!\left(\!\!\binom{x_{12}^\ast}{\boldsymbol{\tilde{x}}}\!\!\right)^{\!\!\top} \overline{\eta}(x_{12}^\ast,\D\boldsymbol{\tilde{x}})
		\intertext{Now we want to scale this orbit to an unit sphere and look only at the last $k-1$ components. Let $\mu$ be the uniform distribution on the sphere $\mathbb{S}_{k-2}$. Let $\boldsymbol{\tilde{f}}$ be the $(k-1)$-dimensional analogue of $\boldsymbol{f}$ and 
					\[\boldsymbol{Q}= \left(\begin{array}{cc}
																		\begin{array}{c} 1\\ x_{12}^\ast \end{array} & \mathbb{O}_{2\times(k-1)}\\
																		\mathbb{O}_{k-1} & \sqrt{1-x_{12}^{\ast\ 2}}\ \mathbb{I}_{k-1}							
																	\end{array} \right)\ .
					\]
					So we have $\boldsymbol{f}((x_{12}^\ast,\boldsymbol{\tilde{x}}\!^\top)\!^\top) = \boldsymbol{Q}\boldsymbol{\tilde{f}}(\boldsymbol{\stackrel{\text{\scalebox{0.6}{$\boldsymbol{\approx}$}}}{x}})$ where $\boldsymbol{\stackrel{\text{\scalebox{0.6}{$\boldsymbol{\approx}$}}}{x}}= \frac{1}{\sqrt{1-x_{12}^{\ast\ 2}}}\,\boldsymbol{\tilde{x}}\in\mathbb{S}_{k-2}$ and}
		\boldsymbol{M}(\xi^\ast)
		&=\frac{q(1)}{k+1} \,  \boldsymbol{f}(\boldsymbol{s}_0) \boldsymbol{f}(\boldsymbol{s}_0)^\top +
			\frac{k\, q(x_{12}^\ast)}{k+1} \, \boldsymbol{Q} \int_{\mathbb{S}_{k-2}} \boldsymbol{\tilde{f}}(\boldsymbol{\stackrel{\text{\scalebox{0.6}{$\boldsymbol{\approx}$}}}{x}}) \boldsymbol{\tilde{f}}(\boldsymbol{\stackrel{\text{\scalebox{0.6}{$\boldsymbol{\approx}$}}}{x}})^\top \mu(\D\!\boldsymbol{\stackrel{\text{\scalebox{0.6}{$\boldsymbol{\approx}$}}}{x}})\ \boldsymbol{Q}^\top
		\intertext{Let $\boldsymbol{\tilde{s}}_1,\ldots,\boldsymbol{\tilde{s}}_{k}$ be the vertices of the (arbitrarily rotated) simplex in Theorem~\ref{theorem:lin1} in the $(k-1)$-dimensional issue and $\boldsymbol{s}_\kappa=\sqrt{1-x_{12}^{\ast\ 2}}\,\boldsymbol{\tilde{s}}_{\kappa}$, $\kappa\in\{1,\ldots,k\}$, the scaled vertices. Using Lemma~\ref{lem:lin1} and Theorem~\ref{theorem:lin1} we get finally}
		\boldsymbol{M}(\xi^\ast)
		&=\frac{1}{k+1} \, q(1) \boldsymbol{f}(\boldsymbol{s}_0) \boldsymbol{f}(\boldsymbol{s}_0)^\top +
			\frac{k}{k+1} \, q(x_{12}^\ast)\, \boldsymbol{Q}\, \frac{1}{k} \sum\limits_{\kappa=1}^k \boldsymbol{\tilde{f}}\!\left(\boldsymbol{\tilde{s}}_\kappa\right) \boldsymbol{\tilde{f}}\!\left(\boldsymbol{\tilde{s}}_\kappa\right)^\top \ \boldsymbol{Q}^\top\\
		&=\frac{1}{k+1} \left[ q(1) \boldsymbol{f}(\boldsymbol{s}_0) \boldsymbol{f}(\boldsymbol{s}_0)^\top +
														\sum\limits_{\kappa=1}^k q(x_{12}^\ast)\, \boldsymbol{f}\!\left(\!\binom{x_{12}^\ast}{\boldsymbol{s}_\kappa}\!\right) \boldsymbol{f}\!\left(\!\binom{x_{12}^\ast}{\boldsymbol{s}_\kappa}\!\right)^\top\right]\ .
	\end{align*}	
	This means that the constructed discrete design has the same information matrix as the optimal (generalised) design. So it is $D$-optimal, too.  
\end{proof}

Now we want to have a short look at some remarks.

\begin{remark} \label{rem1}
	If $\beta_1=0$ in~\eqref{eq:speccase}, then in Lemma~\ref{lem3} the $D$-criterion is invariant with respect to arbitrary rotations of $x_1,x_2,\ldots,x_k$. So we get as in Lemma~\ref{lem4} that the invariant design~$\xi$ is a design with an uniform distribution on the entire sphere $\mathbb{S}_{k-1}$. $\lambda(\boldsymbol{f}(\boldsymbol{x})^\top\boldsymbol{\beta})$ is constant and the information matrix does not depend on it. As in Theorem~\ref{theorem:lin1} and Lemma~\ref{lem:lin1} the discretised design consists of the equally weighted vertices of a regular simplex inscribed in that sphere. The orientation is arbitrary.\hfill$\Diamond$
\end{remark}

\begin{remark} \label{rem2}
	If we want to determine an optimal design for a general $\boldsymbol{\beta}=(\beta_0,\beta_1,\ldots,\beta_k)$, which does not satisfy~\eqref{eq:speccase}, we have to rotate the design and parameter space, find a (locally) $D$-optimal design and rotate it back. By using Theorem~\ref{Theorem1} and an applicable representation of a $(k-1)$-dimensional regular simplex the equally weighted support points are $\frac{\boldsymbol{\tilde{\beta}}}{|\boldsymbol{\tilde{\beta}}|}$, which is the point of the maximum value of the intensity function all over the design region $\mathbb{B}_k$, and the columns of
\[\frac{\boldsymbol{\tilde{\beta}}}{|\boldsymbol{\tilde{\beta}}|} \mathds{1}_k^\top \left(x_{12}^\ast+ \frac{\sqrt{1-x_{12}^{\ast\ 2}}}{\sqrt{k-1}}\right) + \boldsymbol{H}\ \frac{\sqrt{k}}{\sqrt{k-1}}\ \sqrt{1-x_{12}^{\ast\ 2}}\]
where $\boldsymbol{\tilde{\beta}}=(\beta_1,\ldots,\beta_k)$. The second summand contains beside some scaling factor the Householder matrix
\[\boldsymbol{H}=\mathbb{I}_k-\frac{2}{\boldsymbol{v}^\top\boldsymbol{v}}\boldsymbol{v}\boldsymbol{v}^\top \text{\quad and \quad} 
		\boldsymbol{v}=\frac{\boldsymbol{\tilde{\beta}}}{|\boldsymbol{\tilde{\beta}}|}+\frac{1}{\sqrt{k}}(1,\ldots,1)^\top\]
which tilts the standard sub-simplex $(1,0,\ldots,0)^\top,$ $(0,1,\ldots,0)^\top,$ $\ldots,$ $(0,0,\ldots,1)^\top$ in the correct angle (hyperplane). The first summand shifts this in the right distance to $\frac{\boldsymbol{\tilde{\beta}}}{|\boldsymbol{\tilde{\beta}}|}$.\hfill$\Diamond$
\end{remark}

\setcounter{equation}{0} %-1
%\noindent {\bf 2. The Second Section}
\section{Examples}
\label{sec3}
 As shown in Example~\ref{ex1} the Poisson regression and the negative binomial regression satisfy the conditions (A1) - (A4).
  In case of the Poisson regression we can find an explicit solution for $x_{12}^\ast$. Using Lemma~\ref{L9} we have to solve
	\begin{align}
		\beta_1&=\frac{2\,(1+k x_{12}^\ast)}{k\,(1-x_{12}^{\ast\ 2})} \qquad\text{for}\qquad k\geq 2 \label{eq:PoiReg1}
		\vspace*{-2ex}
		\intertext{ and get \vspace*{-2ex}}  
		x_{12}^\ast&=\begin{cases}
										\frac{-1+\sqrt{1-\frac{2}{k}\beta_1+\beta_1^2}}{\beta_1} & \text{for $\beta_1>0$}\\
										-\frac{1}{k}& \text{for $\beta_1=0$\ .}
									\end{cases} \notag%\\
		%k=2: x_{12}^\ast&=\frac{-1+\sqrt{1-\beta_1+\beta_1^2}}{\beta_1}
	\end{align}
	Note that this is right-continuous in $\beta_1=0$. The case $\beta_1=0$ is in line with Remark~\ref{rem1}. 
	For $k=1$, which means Poisson regression on the line segment $[-1,1]$, it is 
	\begin{equation*}
		x_{12}^\ast=\begin{cases}-1& \text{for}\ \beta_1\in[0,1]\\ 1-\tfrac{2}{\beta_1}& \text{for}\ \beta_1>1\end{cases}	
	\end{equation*}
	which is the same solution as in \cite{Russell:2009} or in \cite{Schmidt:2017}.
		
	\noindent For $k=2, 3$ or 4 and higher dimensions the graph of $x_{12}^\ast$ in dependence on $\beta_1\geq 0$ for fixed $k$ is strictly monotonically increasing and reaches 1 for $\beta_1\to\infty$. These graphs approaches monotonically for $k\to\infty$ from below the limit 
	\[ _\infty x_{12}^\ast =\frac{-1+\sqrt{1+\beta_1^2}}{\beta_1}\ \qquad\text{for}\qquad \beta_1>0\]
	and $_\infty x_{12}^\ast=0$ for $\beta_1=0$. For an illustration of these issues see Figure~\ref{fig:plot_beta_x12_poisson_1b}.%\bigskip
	
	\begin{figure}[bt!]
		\centering
			\includegraphics[width=0.45\textwidth]{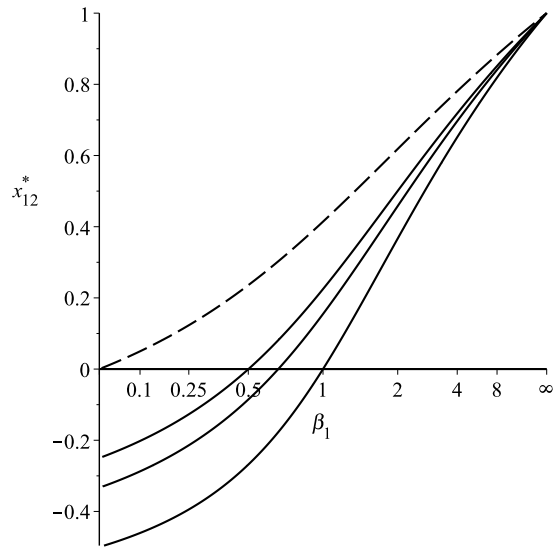}
		\caption{Dependence of $x_{12}^\ast$ on $\beta_1$ for the Poisson regression. The graphs for $k=2$, $k=3$ and $k=4$ (from below) and the graph for $k\to\infty$ (dashed) are plotted. The $\beta_1$-axis is transformed by $\frac{\beta_1}{1+\beta_1}$.}
		\label{fig:plot_beta_x12_poisson_1b}
	\end{figure}
	
	\noindent Using Remark~\ref{rem2} for the Poisson regression on the 3-dimensional unit ball with initial parameter $\boldsymbol{\beta}=(\cdot,1,2,2)^\top$ we get $x_{12}^\ast\approx0.6095$ and a (locally) $D$-optimal design with 4 support points: $(\frac{1}{3},\frac{2}{3},\frac{2}{3})$, $(0.9506,0.2195,0.2195)$, $(-0.1706,0.9852,0.0143)$ and $(-0.1706,0.0143,0.9852)$, see Figure~\ref{fig:PoissonReg_1-2-2_1a}. As implicitly written in the given parameter vector $\boldsymbol{\beta}$ the optimal design does not depend on the first component $\beta_0$.\bigskip
	
	\begin{figure}[bt!]
		\centering
			\includegraphics[width=0.45\textwidth]{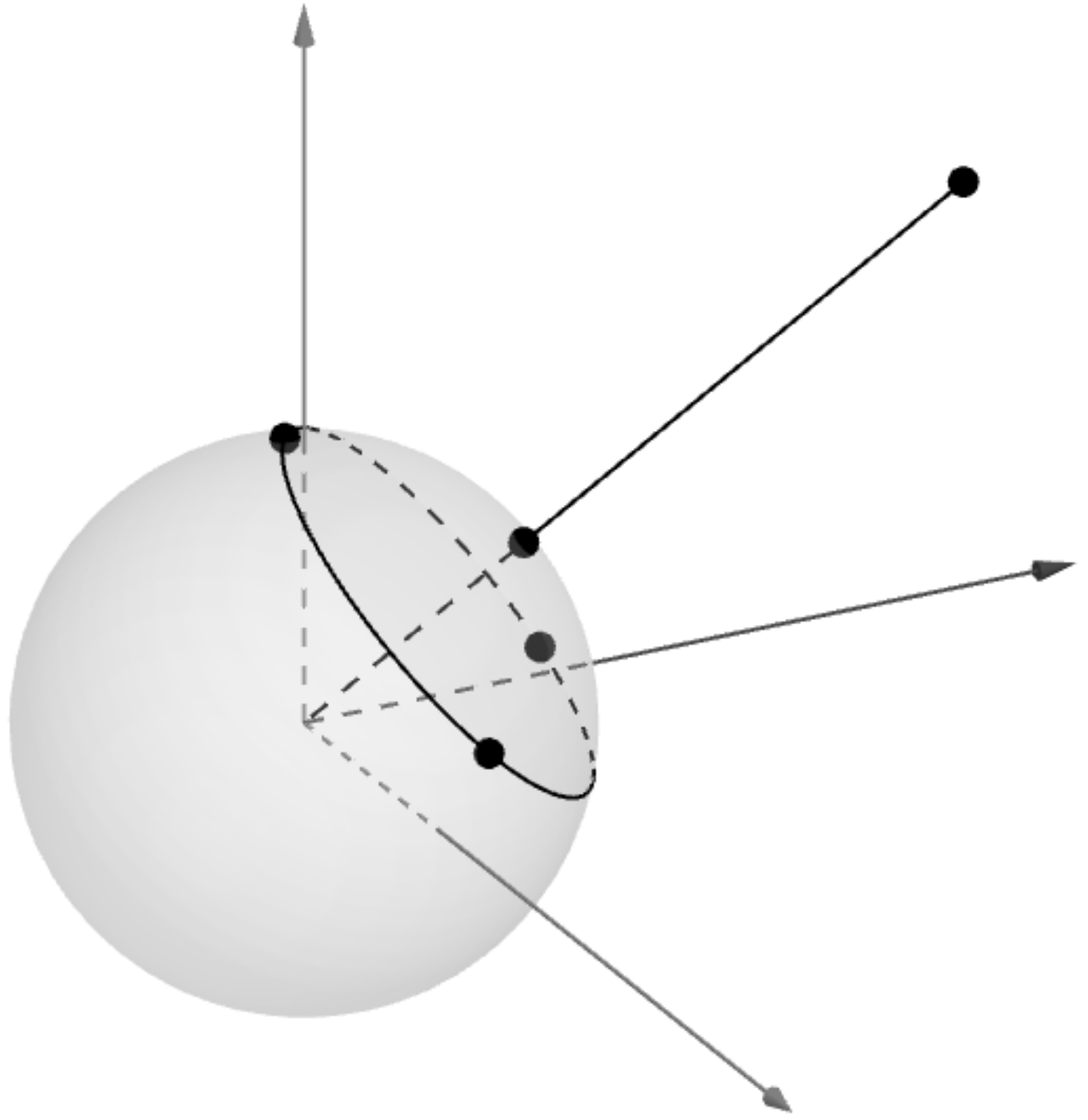}
		\caption{Poisson regression on the 3-dimensional unit ball with initial $\boldsymbol{\beta}=(\cdot,1,2,2)$}
		\label{fig:PoissonReg_1-2-2_1a}
	\end{figure}
	
\newpage	
	In the case of negative binomial regression the solution of~\eqref{eq:L9}, in fact
	\begin{equation}
		\frac{\beta_1}{1+a\exp(\beta_0+\beta_1 x_{12}^\ast)}=\frac{2\,(1+k x_{12}^\ast)}{k\,(1-x_{12}^{\ast\ 2})} \qquad\text{for}\qquad k\geq 2,
	\label{eq:negBin1}
	\end{equation}
	cannot be expressed explicitly. % Well, maybe there is something possible with generalised Lambert-$\operatorname{W}$ functions because of the polynomial structure of $x_{12}^\ast$ with degree~2. However, w
	We plotted the graph only numerically, see Figure~\ref{fig:plot_beta_x12_negbinom_1b}.
	But by the Implicit Function Theorem we know that the $\beta_1$-$x_{12}^\ast$-graph is continuous and differentiable on $(0,\infty)$.
	Additionally there is only one root in $\beta_1=\frac{2}{k} \left(1+a\exp(\beta_0)\right)$ and the derivative of this graph on this root is positive. For $\beta_1\to 0$ the graph approaches $-\frac{1}{k}$. So $x_{12}^\ast$ to the right of the root must be positive.\\	
	The other way round the $x_{12}^\ast$-$\beta_1$-graph can be expressed by the standard Lambert-$\operatorname{W}$ function, the inverse function of $x\mapsto x\exp(x)$:
	\begin{equation*}
		%x_{12}^\ast\mapsto
		-\frac{1}{x_{12}^\ast} 
			\operatorname{W}\!\left(
					- \frac{2\,a\,x_{12}^\ast (1+k x_{12}^\ast)}{k\,(1- x_{12}^{\ast\ 2})}
					\exp\!\left(\beta_0+\frac{2\,x_{12}^\ast(1+k x_{12}^{\ast\ 2})}{k\,(1-x_{12}^{\ast\ 2}) }\right)
			\right)
		+\frac{2\,(1+k x_{12}^\ast)}{k\,(1-x_{12}^{\ast\ 2})}
	%\label{eq:negBin2}
	\end{equation*}
	Here one can see the reason of the behaviour of the actual $\beta_1$-$x_{12}^\ast$-graph. The main branch of the Lambert-$\operatorname{W}$ function induces the shape to the left and the lower branch to the right of the maximum of the curve.\\
	In~\eqref{eq:negBin1} %or in \eqref{eq:negBin2}
	it is obvious, that for $a\to 0$ we get the same graph equation~\eqref{eq:PoiReg1} on $\beta_1\in[0,\infty)$ as in Poisson regression. For $a\to\infty$ we get on $\beta_1\in[0,\infty)$ the constant $x_{12}^\ast=-\frac{1}{k}$.

	\bigskip
	
%\begin{figure}[tb!]
	%\centering
		%\includegraphics[width=0.45\textwidth]{plot_beta_x12_negbinom_1b.pdf}
	%\caption{Dependence of $x_{12}^\ast$ on $\beta_1$ for the negative binomial regression. The graphs for $k=2$, $k=3$ and $k=4$ (from below) and the graph for $k\to\infty$ (dashed) are plotted with $\beta_0=0$ and $a=2$.\newline The $\beta_1$-axis is transformed by $\frac{\beta_1}{1+\beta_1}$.}
	%\label{fig:plot_beta_x12_negbinom_1b}
%\end{figure}

	For proportional hazards models \cite{Schmidt:2017} as well as \cite{Konstantinou:2014} considered particularly three cases of censoring which satisfy our four conditions (A1) - (A4). First of all type I censoring for a fixed censoring time $c$ has 
	\[q_\mathrm{CI}(x_1)=1-\exp(-c \exp(\beta_0+\beta_1 x_1))\ .\]
	Then random censoring with on the intervall $[0,c]$ uniformly distributed censoring times has 
	\[q_\mathrm{CU}(x_1)=1-\frac{\exp(-c \exp(\beta_0+\beta_1 x_1))-1}{c \exp(\beta_0+\beta_1 x_1)}\ .\]
	For an illustration of both issues see Figure~\ref{fig:plot_beta_x12_censoring}.\\
	And finally random censoring with exponentially-$\lambda$-distributed censoring times has
	\[q_\mathrm{CE}(x_1)=\frac{\exp(\beta_0+\beta_1 x_1)}{\exp(\beta_0+\beta_1 x_1)+\lambda}\]
	which is similar to the negative binomial regression with $a=\frac{1}{\lambda}$.%\enlargethispage{\baselineskip}
	
	%\begin{figure} [htb!]
    %%\subfigure[type I censoring with censoring time $c=1$]{\includegraphics[width=0.45\textwidth]{plot_beta_x12_censortypeI_1b.pdf}}
    %%\subfigure[random censoring with ${[0,1]}$-uniformly distributed censoring times]{\includegraphics[width=0.45\textwidth]{plot_beta_x12_censoruniform_1b.pdf}}
    %\hspace*{0.025\textwidth}
		%\subfigure{\includegraphics[width=0.45\textwidth]{plot_beta_x12_censortypeI_1b.pdf}}
    %\hspace*{0.025\textwidth}
    %\subfigure{\includegraphics[width=0.45\textwidth]{plot_beta_x12_censoruniform_1b.pdf}}
%\caption{Dependence of $x_{12}^\ast$ on $\beta_1$ for the type I censoring with censoring time $c=1$ (left plot) and for the random censoring with $[0,1]$-uniformly distributed censoring times (right plot). The graphs for $k=2$, $k=3$ and $k=4$ (from below) and the graph for $k\to\infty$ (dashed) are plotted with $\beta_0=0$. The $\beta_1$-axis is transformed by $\frac{\beta_1}{1+\beta_1}$.}
	%\label{fig:plot_beta_x12_censoring}
%\end{figure}
 
\newpage	
\begin{figure}[tb!]
	\centering
		\includegraphics[width=0.45\textwidth]{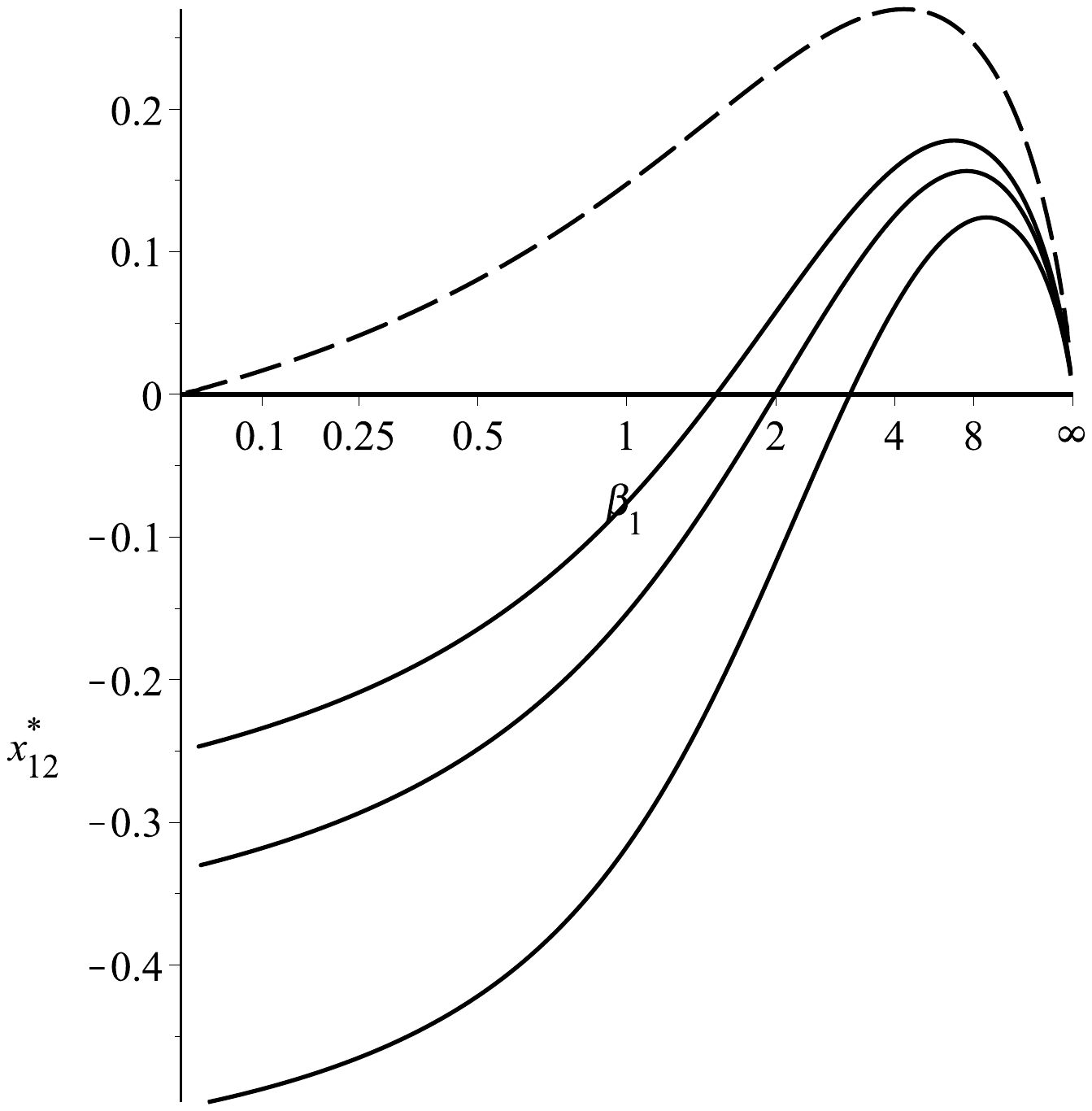}
	\caption{Dependence of $x_{12}^\ast$ on $\beta_1$ for the negative binomial regression. The graphs for $k=2$, $k=3$ and $k=4$ (from below) and the graph for $k\to\infty$ (dashed) are plotted with $\beta_0=0$ and $a=2$.\newline The $\beta_1$-axis is transformed by $\frac{\beta_1}{1+\beta_1}$.}
	\label{fig:plot_beta_x12_negbinom_1b}
\end{figure}	

	\begin{figure} [htb!]
    %\subfigure[type I censoring with censoring time $c=1$]{\includegraphics[width=0.45\textwidth]{plot_beta_x12_censortypeI_1b.pdf}}
    %\subfigure[random censoring with ${[0,1]}$-uniformly distributed censoring times]{\includegraphics[width=0.45\textwidth]{plot_beta_x12_censoruniform_1b.pdf}}
    \hspace*{0.025\textwidth}
		\subfigure{\includegraphics[width=0.45\textwidth]{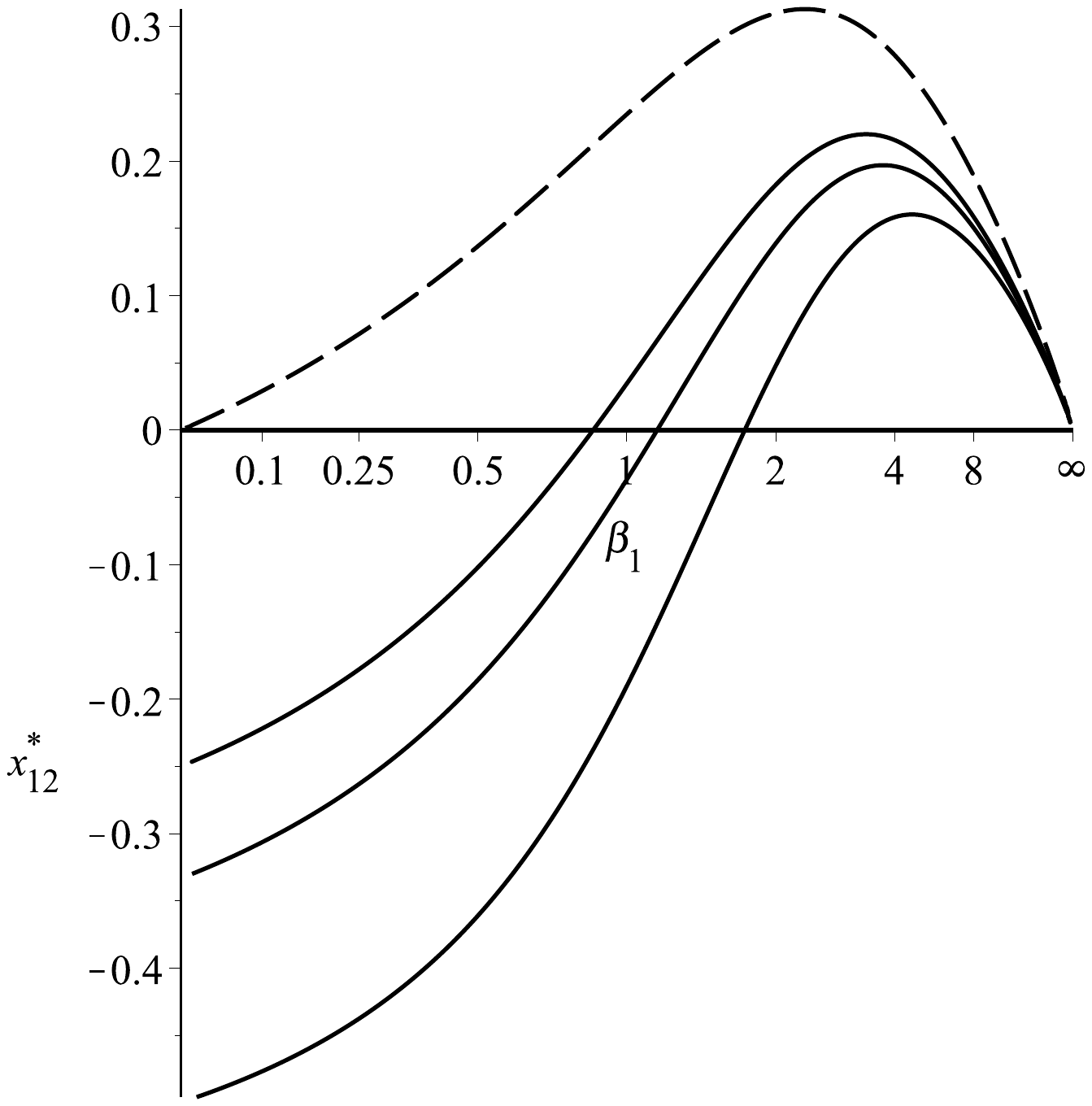}}
    \hspace*{0.025\textwidth}
    \subfigure{\includegraphics[width=0.45\textwidth]{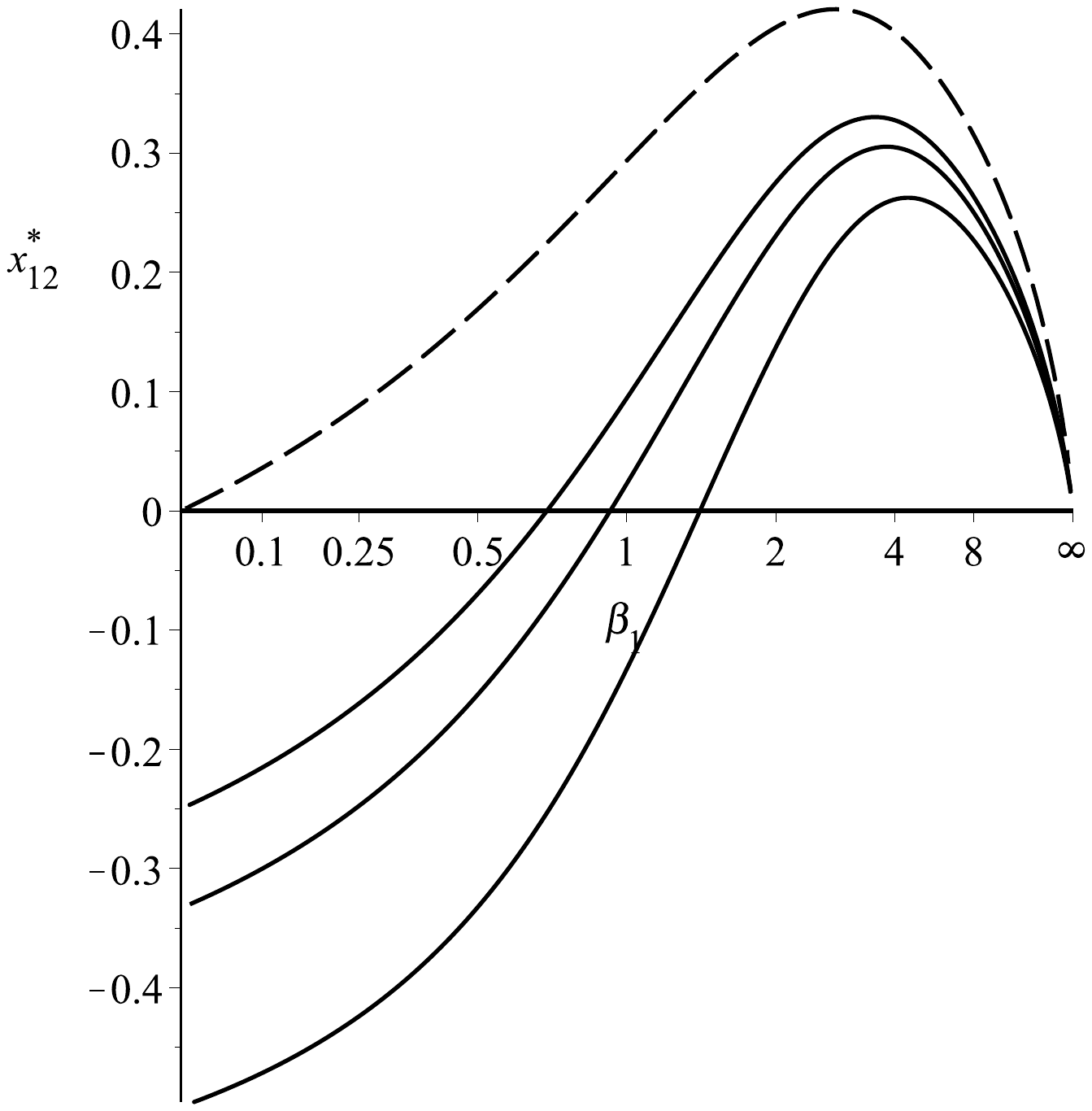}}
\caption{Dependence of $x_{12}^\ast$ on $\beta_1$ for the type I censoring with censoring time $c=1$ (left plot) and for the random censoring with $[0,1]$-uniformly distributed censoring times (right plot). The graphs for $k=2$, $k=3$ and $k=4$ (from below) and the graph for $k\to\infty$ (dashed) are plotted with $\beta_0=0$. The $\beta_1$-axis is transformed by $\frac{\beta_1}{1+\beta_1}$.}
	\label{fig:plot_beta_x12_censoring}
\end{figure}

\newpage

\setcounter{equation}{0} %-1
%\noindent {\bf 2. The Second Section}
\section{Summary}
\label{sec4}

The developed (locally) $D$-optimal exact designs are applicable to a large class of non-linear multiple regression problems, such as Poisson or negative binomial regression and censoring models. The design region is the $k$-dimensional ball. By our opinion there is a big amount of practical real-life problems for example in engineering or physics where the phenomenon has to be examined and evaluated only in a ball-shaped neighbourhood of its occurrence.\\
Not only the (unit) ball is a possible design region but also other regions which can be linearly transformed by scaling and rotations into a ball can be treated by using the equivariance results in \cite{Radloff:2016}. Any $k$-dimensional ball with arbitrary radius or an ellipsoid are such solids.\\
The results are only valid for multiple regression. If there are interaction terms or quadratic terms, support points inside the design region or more complicated designs will occur.\\
In this paper we focused only on (local) $D$-optimality. There is a dependence on the initial parameter vector. Other optimality criteria or especially (more) robust criteria, like maximin, and their designs on the $k$-dimensional unit ball should be the object of future research.

%%%%%%%%%%%%%%%%%%%%%%%%%%%%%%%%%%%%%%%%%%%%%%%%%%%%%%%%%%%%%%%%%%%%%%%%%%%%%%%%%%%%%%%%%%%%%%%%%%%%%%%%%%%%%%%%%%%%%%%%%%%%

%\iffalse
\bibhang=1.7pc
\bibsep=2pt
\fontsize{9}{14pt plus.8pt minus .6pt}\selectfont
\renewcommand\bibname{\large \bf References}
%\begin{thebibliography}{11}
\expandafter\ifx\csname
natexlab\endcsname\relax\def\natexlab#1{#1}\fi
\expandafter\ifx\csname url\endcsname\relax
  \def\url#1{\texttt{#1}}\fi
\expandafter\ifx\csname urlprefix\endcsname\relax\def\urlprefix{URL}\fi
%\fi

\bibliographystyle{chicago}      % Chicago style, author-year citations
\bibliography{BibTeX_Radloff_1}   % name your BibTeX data base

\end{document}